\documentclass[letterpaper, 10 pt, conference]{ieeeconf}  

\IEEEoverridecommandlockouts                              

\usepackage{graphics} 
\usepackage{amsmath} 
\usepackage{amssymb}  
\usepackage{enumerate}
\usepackage{mathtools}
\usepackage{graphicx}
\usepackage{graphicx,lipsum}
\usepackage[noadjust]{cite}

\newtheorem{thm}{\textbf{Theorem}}
\newtheorem{lem}{\textbf{Lemma}}
\newtheorem{prop}{\textbf{Proposition}}
\newtheorem{cor}{\textbf{Corollary}}
\newtheorem{defn}{\textbf{Definition}}

\newtheorem{assum}{\textbf{Assumption}}

\newcommand{\bm}[1]{\boldsymbol{#1}}
\DeclareMathOperator{\diag}{diag}
\DeclareMathOperator{\sgn}{sgn}

\title{\LARGE \bf
Learning Piecewise Residuals of Control Barrier Functions for Safety of Switching Systems using Multi-Output Gaussian Processes
}

\author{Mohammad Aali and Jun Liu
\thanks{This work was supported in part by Nutrien Ltd., NSERC of Canada, and the Canada Research Chairs program.}
\thanks{Mohammad Aali is with the Department of Applied Mathematics, University of Waterloo, Canada,
        {\tt\small mohammad.aali@uwaterloo.ca}}%
\thanks{Jun Liu is with the Department of Applied Mathematics, University of Waterloo, Canada,
        {\tt\small j.liu@uwaterloo.ca}}%
}

\begin{document}

\maketitle
\thispagestyle{empty}
\pagestyle{empty}

\begin{abstract}
Control barrier functions (CBFs) have recently been introduced as a systematic tool to ensure safety by establishing set invariance. When combined with a control Lyapunov function (CLF), they form a safety-critical control mechanism. However, the effectiveness of CBFs and CLFs is closely tied to the system model. In practice, model uncertainty can jeopardize safety and stability guarantees and may lead to undesirable performance. In this paper, we develop a safe learning-based control strategy for switching systems in the face of uncertainty. We focus on the case that a nominal model is available for a true underlying switching system. This uncertainty results in piecewise residuals for each switching surface, impacting the CLF and CBF constraints. We introduce a batch multi-output Gaussian process (MOGP) framework to approximate these piecewise residuals, thereby mitigating the adverse effects of uncertainty. A particular structure of the covariance function enables us to convert the MOGP-based chance constraints CLF and CBF into second-order cone constraints, which leads to a convex optimization. We analyze the feasibility of the resulting optimization and provide the necessary and sufficient conditions for feasibility. The effectiveness of the proposed strategy is validated through a simulation of a switching adaptive cruise control system.
\end{abstract}
\section{Introduction}\label{sec1}
Many real-world control problems are modeled as hybrid dynamical systems, which involve a coupling between continuous dynamics and discrete events. Research in the stability analysis of such systems has grown significantly in recent decades \cite{branicky1998multiple, daafouz2002stability, liu2011input}. More recently, Gaussian processes (GPs) are used to approximate hybrid residual dynamics model which can be integrated into a stabilizing model predictive control design ensuring probabilistic constraint satisfaction \cite{d2023stochastic}. However, in the modern control applications, the focus has expanded from mere stability assurance to encompass safety verification. CBFs have recently provided a powerful theoretical tool for synthesizing controllers that ensure the safety of dynamical systems. They were initially applied to autonomous driving and bipedal walking \cite{ames2019control}. Stability guarantees are then added via a unified control Lyapunov function-control barrier function (CLF-CBF) framework by solving a state-dependent quadratic programming (QP) problem at each time step to compute the control input \cite{romdlony2016stabilization}.
In the context of hybrid systems, barrier functions were initially used as safety certificates in \cite{prajna2004safety}. After the introduction of CBFs, CLF-CBF design has been extensively used in safety-critical methods for hybrid systems \cite{ames2019control, marley2021synergistic, meng2023lyapunov}. However, there are very few studies on learning-based safety verification of switching systems via CBFs in the presence of uncertainty.

Learning-based CBFs for uncertain systems can be divided into two approaches. The first approach learns the unknown dynamics and uses the approximated model to derive safety certificates. In \cite{wang2018safe}, a GP is employed to learn the unmodeled dynamics for the safe navigation of quadrotors with CBFs. A GP-based systematic controller synthesis method is introduced in \cite{jagtap2020control}, offering an inherently safe control strategy. In \cite{fan2020bayesian}, the authors employ a Bayesian approach to learn model uncertainty while ensuring safety during the learning process. Learning the complete system dynamics with GPs is computationally demanding. Therefore, a common approach is to learn each component of the vector fields individually, disregarding their correlations. Furthermore, existing studies typically neglect uncertainty in the actuation term to simplify the learning-based model.

These limitations lead us to the second approach, which approximates the projection of the residual onto the CLF and CBF. In \cite{csomay2021episodic}, an episodic learning method using neural networks is employed to approximate the impact of unmodeled dynamics of hybrid bi-pedal robot on system safety, eliminating the need for exhaustive data collection. However, methods employing neural networks must achieve precise approximations of the dynamics to formally ensure safety. In general, this requirement is challenging to meet, but GPs can provide uncertainty quantification with analytical confidence bounds. In \cite{castaneda2021pointwise}, a GP-based min-norm controller stabilizing an unknown control affine system while guaranteeing safety using CLF-CBF method is introduced. The effect of the uncertainty in the stability and safety constraints is approximated by a GP model. This technique results in the approximation of a real-valued function instead of a vector field, leading to a significant reduction in the model complexity for high-dimensional systems. In \cite{aali2024learning}, the effect of uncertainty on high-order CBFs is quantified using GPs.

In this paper, we develop a batch MOGP model to learn the piecewise residuals in each switching surface. We use the multi-output design to learn the piecewise residuals corresponding to CLF and CBF simultaneously, and the batched design facilitates for efficient computation of the model's hyperparameters.
The main contributions of this paper relative to prior works are presented as follows:
\begin{itemize}
\item We show that the switching system imposes piecewise residuals on CLF and CBF constraints. We develop a batch MOGP model to efficiently approximate residuals in each switching surface.
\item We demonstrate that by selecting a particular form of the covariance function in the batch MOGP structure, the resulting min-norm controller with an uncertainty-aware chance constraint can be converted into a second-order cone program (SOCP).
\item We analyze the feasibility of the resulting constrained optimization problem and determine the necessary and sufficient conditions for feasibility.
\end{itemize}
A conference version of this paper will appear in \cite{ECC2024}. This extended version includes additional details omitted in the conference paper.
\section{Problem setup and statement}\label{sec2}
Consider the control-affine switching system of the form
\begin{align}
    \dot {\bm{x}} &= \sum^R_{r=1} \delta_r \left( f_r(\bm{x})+ g_r(\bm{x}) \bm{u} \right),\label{eq1}\\
    \delta _{r} &= \begin{cases} {1}&{{\text{if}}\quad \bm x \in \mathcal{R}_r,} \\ {0}&{{\text{otherwise}},} \end{cases} \nonumber
    \label{eq1}
\end{align}where  $\boldsymbol{x} \in \mathcal{X} \subset \mathbb{R}^n$ is the state, and $\boldsymbol{u} \in \mathbb{R}^m$ is the input, and $\mathcal{X}$ is a compact and convex set. The function $\delta_r:\mathcal{X} \rightarrow \{0, 1\}, r \in \mathcal{I}$ is a state-dependent switching signal with the finite index set $\mathcal{I} = \{1, \dots, R \}$ which indicates the index of the active switching surface $\mathcal{R}_r \subset \mathcal{X}$. The vector fields $f_r:\mathbb{R}^n \rightarrow \mathbb{R}^n$ and $g_r:\mathbb{R}^n \rightarrow \mathbb{R}^{m \times n}, r \in \mathcal{I}$, are the unknown dynamics, which are assumed to be locally Lipschitz in their arguments.
Let $\pi:\mathbb{R}^n \rightarrow \mathbb{R}^{m}$ be a locally Lipschitz continuous state feedback control law. We assume that the impulse effects are absent, i.e., the reset map is the identity. Furthermore, we assume that the closed-loop system
\begin{equation}
    \dot {\boldsymbol{x}} = F_{cl}(\boldsymbol{x}) \triangleq \sum^R_{r=1} \delta_r\left( f_r(\bm{x})+ g_r(\bm{x}) \pi(\bm x) \right),
    \label{eq2}
\end{equation}
satisfies Caratheodory condition \cite{hale2009ordinary}. 
Then, for any initial state $\boldsymbol{x}(t_0) = \boldsymbol{x}_0$, system (\ref{eq2}) admits a solution $\boldsymbol{x}(t)$ in the sense of Caratheodory defined on a maximal interval of existence $I(\boldsymbol{x}_0) = [t_0, I_{max})$. In this paper, we consider complete solutions, thus $I_{max} = \infty$.
\begin{assum}\label{assum1}
    We assume that the state space is partitioned into $R$ non-overlapping regions labeled as $\mathcal{R}_r \subset \mathcal{X}$, such that they cover the whole state space, i.e., $\cup_{r=1}^R \mathcal{R}_r = \mathcal{X}$, and $\mathcal{R}_i \cap \mathcal{R}_j = \emptyset$ for all $i,j \in \{1, \dots, R \}, i \neq j$.
\end{assum}
Assumption \ref{assum1} requires system (\ref{eq1}) to be well-posed, i.e., for all $\bm x \in \mathcal{X}$ there exists only one active index $r \in \mathcal{I}$ satisfying the membership condition.
\subsection{Control Lyapunov Functions for Switching Systems}\label{sec2a}
\begin{defn}[Class $\mathcal{K}$ function \cite{khalil2002nonlinear}]\label{def1}
    We say that a continuous function $\alpha:[0, a) \rightarrow [0, \infty), a>0$, belongs to class $\mathcal{K}$, if it is strictly increasing and $\alpha(0)=0$.
\end{defn}

\begin{defn}[Control Lyapunov functions]\label{def2}
    A continuously differentiable function $V:\mathcal{X} \rightarrow \mathbb{R}$ is a CLF for system (\ref{eq1}) if there exist positive constants $\lambda, c_1, c_2 > 0$ such that for all $\bm x \in \mathcal{X}$ and $r \in \mathcal{I}$,
    \begin{align}
        &c_1 \|\bm x \|^2 \leq V(\bm x) \leq c_2 \|x \|^2,\nonumber\\
        &\sum^R_{r=1} \delta_r \left( L_{f_r} V(\bm x)+L_{g_r} V(\bm x) \bm u \right ) \leq -\lambda V(\bm x),
    \label{eq3}
    \end{align}
    where $L_{f_r} V(\bm x)$ and $L_{g_r} V(\bm x)$ are the Lie derivatives of $V(\bm x)$ with respect to the corresponding vector fields.
\end{defn}
Therefore, if the subsystems in (\ref{eq1}) share a CLF, the rate of decrease of $V$ along the vector fields, given by (\ref{eq3}), is not affected by the switching, thus exponential stability is uniform with respect to $\delta_r$. Also, if there exists a compact subset $D \subseteq \mathcal{X}$ including the origin such that (\ref{eq3}) holds for some $\lambda > 0$, and there exists a sublevel set $\Omega = \{ \bm x \in D \mid V(\bm x) \leq c \}$, where $c > 0$, then the origin is locally exponentially stabilizable from $\Omega$.
\subsection{Control Barrier Functions}
CBF method defines safety based on the notion of set invariance, where a subset of the state space is specified as the safe set. This set is characterized by the zero-superlevel set of a continuously differentiable function $h:\mathbb{R}^n \rightarrow \mathbb{R}$ as
\begin{align}
    \mathcal{C} &= \{\boldsymbol{x} \in \mathbb{R}^n \mid h(\boldsymbol{x})\geq 0\},\nonumber\\
    \partial \mathcal{C} &= \{\boldsymbol{x} \in \mathbb{R}^n \mid h(\boldsymbol{x})= 0\}.\label{eq4}
\end{align}

CBFs provide a constructive tool for achieving forward invariance of set $\mathcal{C}$. We define a CBF as follows:
\begin{defn}[Control barrier function \cite{ames2019control}]\label{def3}
     Given the set $\mathcal{C}$ as defined in (\ref{eq4}), the continuously differentiable function $h(\boldsymbol{x})$ is called a CBF on a domain $\mathcal{D}$ with $\mathcal{C} \subset \mathcal{D} \subset \mathbb{R}^n$, if there exists a class $\mathcal{K}$ function
    $\alpha$ such that the following holds for all $\boldsymbol{x} \in \mathcal{D}$ and $r \in \mathcal{I}$
\begin{equation}
    \sum^R_{r=1} \delta_r \left( L_{f_r} h(\bm x)+L_{g_r} h(\bm x) \bm u \right )  \geq -\alpha(h(\bm x)).
    \label{eq5}
\end{equation}
\end{defn}
We can derive the following Corollary based on Nagumo's theorem \cite{blanchini2008set} for the forward invariance of the safe set $\mathcal{C}$:
\begin{cor}\label{cor1}
     Given CBF $h:\mathbb{R}^n \rightarrow \mathbb{R}$ with the associated set $\mathcal{C}$ in (\ref{eq3}), if $\nabla h(\bm x) \neq 0$ for all $\bm x \in \partial \mathcal{C}$, any Lipschitz continuous controller $\boldsymbol{u}(\boldsymbol{x})$ satisfying (\ref{eq5}), guarantees that the set $\mathcal{C}$ is forward invariant for the system (\ref{eq1}) and thus safe.
\end{cor}
Inequalities (\ref{eq3}) and (\ref{eq5}) impose an affine condition on the control values which can be used to ensure stability and safety.
Consider the CLF constraint (\ref{eq3}), primarily designed to fulfill control objectives. Our goal is to apply the control law that satisfies (\ref{eq3}) to the system only if it complies with (\ref{eq5}). In practice, this is accomplished by solving the following QP optimization:
\begin{subequations}\label{eq6}
    \begin{align}
        \boldsymbol{u}_{s} ={} &\underset{\left( \bm u , d \right) \in \mathcal{R}^{m+1}}{\arg\min} \hspace{4pt} {\| \bm u \|^2_2 + \rho d^2} \label{eq6a}\\
        \textrm{s.t.}\quad&\sum^R_{r=1} \delta_r \left( L_{f_r} V(\bm x)+L_{g_r} V(\bm x) \bm u \right ) + \lambda V(\bm x) \leq d,\label{eq6b}\\
        &\sum^R_{r=1} \delta_r \left( L_{f_r} h(\bm x)+L_{g_r} h(\bm x) \bm u \right )  + \alpha(h(\bm x)) \geq 0, \label{eq6c}
    \end{align}
\end{subequations}
where $\rho$ is a positive coefficient and $d \in \mathbb{R}$ is a slack variable.
The resulting minimally invasive point-wise controller prioritizes safety over control objectives by satisfying (\ref{eq5}) as a hard constraint.

\subsection{Impact of uncertainty in the switching systems}\label{sec3}
In this section, we focus on reformulating CLF and CBF constraints such that they take the model uncertainty into account. We refer to the system (\ref{eq1}) as the true system, which is unknown and modeled by a single nominal model in all regions as
\begin{equation}
    \dot {\boldsymbol{x}} = \hat f(\boldsymbol{x})+ \hat g(\boldsymbol{x}) \boldsymbol{u},
    \label{eq7}
\end{equation}
where $\hat f:\mathbb{R}^n \rightarrow \mathbb{R}^n$ and $\hat g:\mathbb{R}^n \rightarrow \mathbb{R}^{n \times m}$ are locally Lipschitz in $\boldsymbol{x}$. We design a locally exponentially stabilizing CLF and CBF based on the nominal system and assume they are also valid for the true system. 

Due to the inherent model mismatch between the true system and the nominal model, the time derivatives of $V(\bm x)$ and $h(\bm x)$ based on the nominal model differ from their true values. The resulting error can be represented as
\begin{align*}
    d^V(\bm x, \bm u) &= \dot V(\bm x, \bm u) - \hat{\dot V}(\bm x, \bm u) = \sum^R_{r=1} \delta_r d_r^V(\bm x, \bm u),\\
    d^h(\bm x, \bm u) &= \dot h(\bm x, \bm u) - \hat{\dot h}(\bm x, \bm u) = \sum^R_{r=1} \delta_r d_r^h(\bm x, \bm u),
\end{align*}
where for $r \in \mathcal{I}$, we have
\begin{align}
    d_r^V &= L_{f_r}V(\bm x) - L_{\hat f} V(\bm x) + \left(L_{g_r}V(\bm x) - L_{\hat g} V(\bm x)\right) \bm u,\nonumber\\
    d_r^h &= L_{f_r}h(\bm x) - L_{\hat f} h(\bm x) + \left(L_{g_r}h(\bm x) - L_{\hat g} h(\bm x)\right) \bm u,\label{eq8}
\end{align}
where the dependence on $\bm x$ and $\bm u$ is removed for simplicity.
Consequently, these discrepancies propagate into the CLF (\ref{eq6b}) and CBF (\ref{eq6c}) constraints, potentially jeopardizing the stability and safety of the system. Using (\ref{eq8}), the constraints (\ref{eq6b}) and (\ref{eq6c}) become
\begin{align}
    &L_{\hat f} V(\bm x) + L_{\hat g} V(\bm x) \bm u + \sum^R_{r=1} \delta_r d_r^V + \lambda V(\bm x) \leq d,\label{eq8a}\\
    &L_{\hat f} h(\bm x) + L_{\hat g} h(\bm x) \bm u + \sum^R_{r=1} \delta_r d_r^h  + \alpha(h(\bm x)) \geq 0. \label{eq8b}
\end{align}
In this paper, our primary focus is characterizing the adverse effects of uncertainty in the proposed safety-critical control method. We then adopt a supervised learning approach to make our design robust to piecewise residuals of the form (\ref{eq8}). Thus, a dataset that contains information about piecewise residuals in all regions is needed.
\begin{assum}\label{assum2}
    We have access to the set of samples $\{ (\bm x(t), \bm u(t)) \}, t \in [t_0, t_f]$ collected with sampling time $\Delta t$ from trajectories of the system. We assume this set is rich and contains states from all regions $\mathcal{R}_r, r \in \mathcal{I}$.
\end{assum}
Given these data samples, we can approximately measure $d^V$ and $d^h$ by collecting trajectories from the true system and using the finite difference method to obtain
\begin{align} 
\omega_j^V &= \left( {V\left( {\bm x(t_j + \Delta t)} \right) - V\left( {\bm x(t_j)} \right)} \right)/\Delta t - \hat{\dot V}\left( {{\bm x_j},{\bm u_j}} \right), \nonumber\\ 
\omega_j^h &= \left( {h\left( {\bm x(t_j + \Delta t)} \right) - h\left( {\bm x(t_j)} \right)} \right)/\Delta t - \hat{\dot h}\left( {{\bm x_j},{\bm u_j}} \right),\label{eq9}
\end{align}
where $t_j = t_0 + (j-1)\Delta t$, $\bm x_j = \left (\bm x(t_j + \Delta t) + \bm x(t_j) \right)/2$ is the mean of the state, and $\bm u_j = \bm u(t_j)$ is the control input for $t_j \in [t_j, t_j + \Delta t), j \in \{1, \dots, N \}$.

A dataset can be generated by collecting trajectories of the true system. We denote it by $\mathcal{D} = \{ ((\bm x_j, \bm u_j), \omega_j^V, \omega_j^h) \}_{j=1}^N$, where $(\bm x_j, \bm u_j) \in \mathcal{X} \times \mathbb{R}^m$ is the input data and outputs $\omega_j^V$ and $\omega_j^h$ are obtained from (\ref{eq9}).
We will use this dataset to learn piecewise residuals.
\section{Proposed batch mogp-based design for piecewise residuals}
\subsection{Single-output GP for real-valued functions}\label{sec3a}
Typical single-output GP (SOGP) provides a framework to approximate nonlinear functions. We represent an SOGP approximation of $v:\mathcal{X} \rightarrow \mathbb{R}$ as $v(\bm x) \sim \mathcal{GP}\left(m(\bm x), k(\bm x, \bm x') \right)$, which is fully specified by its mean function $m:\mathcal{X} \rightarrow \mathbb{R}$ and covariance (kernel) function $k:\mathcal{X} \times \mathcal{X} \rightarrow \mathbb{R}$. Without loss of generality, we consider zero prior mean. Our prior knowledge of the problem can shape the form of $k(\bm x, \bm x')$, considering that the covariance function must satisfy positive definiteness \cite{williams2006gaussian}. Given noisy measurements $w_j = v(\bm x_j) +  \varepsilon_j, j\in \{1, \dots, N \}$, which are corrupted by Gaussian noise $\varepsilon_j \sim \mathcal{N}(0,\,\sigma^{2})$, a GP model can infer a posterior mean and variance for a test point $\bm x_t$ conditioned on the measurements
\begin{align} 
    \mu(\bm x_t) &= {{\bm{w}}^T}{\left( {K + \sigma _n^2 I} \right)^{-1}}\bar {K}^T,\nonumber \\ 
    \sigma(\bm x_t)^2 &= k\left( {{\bm x_t},{\bm x_t}} \right) - {\bar K}{\left( {K + \sigma _n^2 I} \right)^{- 1}}\bar{K}^T,
    \label{eq10}
\end{align}
where $\bm w \in \mathbb{R}^N$ is the vector of measurements $w_j$, $K\in \mathbb{R}^{N \times N}$ is the Gram matrix with elements $K_{ij} = k(\bm x_i, \bm x_j)$, and $\bar K = \begin{bmatrix}
    k(\bm x_t, \bm x_1), \dots, k(\bm x_t, \bm x_N)
\end{bmatrix}^T \in \mathbb{R}^N$.

\subsection{Batch multi-output GP for piecewise residuals}\label{sec3b}
Going back to the original problem, our objective is to build a Bayesian approximation uncorrelated target functions $d^V$ and $d^h$. 
We need to approximate each individual $d_r^V$ and $d_r^h$, for $r \in \mathcal{I}$. Since the input data $(\bm x_j, \bm u_j)$ gathered from the same trajectory, we leverage from batch multi-output structure to learn $d_r^V$ and $d_r^h$ simultaneously. As $d_r^V$ and $d_r^h$ are uncorrelated, we consider an independent MOGP framework. This design facilitates more computationally efficient GP training.
We initially partition the dataset $\mathcal{D}$, as described in Section \ref{sec2} into $R$ datasets of the form
\begin{equation*}
    \mathcal{D}_r = \{ ((\bm x_j, \bm u_j), \omega_j^V, \omega_j^h)| ((\bm x_j, \bm u_j), \omega_j^V, \omega_j^h) \in \mathcal{D}, \bm x_j \in \mathcal{R}_r  \}
\end{equation*}
for $r \in \mathcal{I}$, and denote the cardinality of each set by $N_r$. In a batch GP, each dataset or batch is associated with its own GP, and these GPs can capture the specific behavior and dependencies within each batch. Now, we employ MOGP to approximate $R$ individual GPs for $d_r^V$ and $d_r^h$ simultaneously within a unified framework. Let $\bm d_r(\bm x, \bm u) = \begin{bmatrix}d_r^V(\bm x, \bm u) & d_r^h(\bm x, \bm u) \end{bmatrix}^T \in \mathbb{R}^2$ be the vector-valued function which follows an MOGP as
\begin{equation}
\bm d_r(\bm x, \bm u) \sim  \mathcal{GP}^r\left ( \begin{bmatrix}0  \\ 0\end{bmatrix}, \begin{bmatrix}
k_r^V & 0 \\ 0 & k_r^h \end{bmatrix}\right ).\label{eq11}
\end{equation}
where the dependence to $(\bm x, \bm u)$ for the covariance functions is removed for simpler notation. As a result, training a MOGP model (\ref{eq11}) leading to separate posterior distributions for target functions $d_r^V$ and $d_r^h$:
\begin{align}
    d_r^V(\bm x, \bm u)  &\sim \mathcal{N}(m_r^V(\mathcal{D}_r, \bm x, \bm u) , {\sigma_r^V}^2(\mathcal{D}_r, \bm x, \bm u)),\nonumber\\
    d_r^h(\bm x, \bm u) &\sim \mathcal{N}(m_r^h(\mathcal{D}_r, \bm x, \bm u), {\sigma_r^h}^2(\mathcal{D}_r, \bm x, \bm u)).\label{eq12}
\end{align}
From (\ref{eq8}), we can verify that the residuals are control-affine and can be characterized by
\begin{align}
    d_r^V &= \begin{bmatrix}
        L_{f_r}V - L_{\hat f} V & L_{g_r}V - L_{\hat g} V
    \end{bmatrix} \begin{bmatrix}
        1\\ \bm u
    \end{bmatrix} = \bm \varphi_r^V \bm y,\label{eq13a}\\
    d_r^h &= \begin{bmatrix}
        L_{f_r}h - L_{\hat f} h & L_{g_r}h - L_{\hat g} h
    \end{bmatrix} \begin{bmatrix}
        1\\ \bm u
    \end{bmatrix}= \bm \varphi_r^h \bm y.\label{eq13b}
\end{align}
We denote the concatenation of one and $\bm u$ by $\bm y = \begin{bmatrix} 1, \bm u^T \end{bmatrix}^T \in \mathbb{R}^{m+1}$ and $\bm \varphi_r^V,\bm \varphi_r^h \in \mathbb{R}^{1 \times m+1}$. This prior knowledge of the problem can be embedded into the structure of the kernel function.

\begin{prop}\label{prop2}
     Let $\bm x \in \mathcal{R}_r$, $\bm y \in \mathbb{R}^{m+1}$, and define the input domain $\mathcal{\bar X}_r = \mathcal{R}_r \times \mathbb{R}^{m+1}$ for a given $r \in \mathcal{I}$. Then, consider the real-valued function $k_r: \mathcal{\bar X}_r \times \mathcal{\bar X}_r \rightarrow \mathbb{R}$, defined by
    \begin{equation}
        k_r \left(\begin{bmatrix}
        \bm x\\\bm y
    \end{bmatrix}, \begin{bmatrix}
        \bm x'\\ \bm y'
    \end{bmatrix} \right ) = \bm y^T \Lambda_r(\bm x, \bm x') \bm y', 
    \label{eq13}
    \end{equation}
    where $\Lambda_r(\bm x, \bm x') = \diag(\begin{bmatrix}
        k_r^1(\bm x,\bm x'), \dots, k_r^{m+1}(\bm x,\bm x')\end{bmatrix} )$ and $k_r^i: \mathcal{R}_r \times \mathcal{R}_r \rightarrow \mathbb{R}, i\in\{1, \dots, m+1 \}$, we denote the individual components $k_r^i$'s as base kernels. If base kernels are all positive-definite, then $k_r$ is a positive-definite kernel.
\end{prop}

\begin{proof}
    Since $k_r^i(\bm x, \bm x')$'s are positive definite kernels, then by the definition, there is a feature map $\varphi(\bm x)$ such that $k_r^i(\bm x, \bm x') = \varphi^T_i(\bm x) \varphi_i(\bm x')$ for $i\in \{1, \dots, m+1 \}$ \cite{bishop2006pattern}. Now, Let $y^i$ and $y'^i$ be the $i^{th}$ element of the corresponding vectors. We have
    \begin{align*}
        k_r &= \bm y^T \diag \left([
        \varphi^T_1(\bm x) \varphi_1(\bm x'), \dots, \varphi^T_{m+1}(\bm x) \varphi_{m+1}(\bm x')
    ] \right) \bm y'\\
    &= \sum_{i=1}^{m+1} y^i \varphi_i^T(\bm x) \varphi_i(\bm x') y'^i\\
    &= \begin{bmatrix}
         y^1 \varphi_1^T(\bm x) & \dots & y^{m+1} \varphi_{m+1}^T(\bm x)
    \end{bmatrix}
    \begin{bmatrix}
        y'^1 \varphi_1(\bm x') \\ \vdots \\ y'^{m+1} \varphi_{m+1}(\bm x')
    \end{bmatrix}\\
        &= \psi^T\left(\begin{bmatrix}
        \bm x\\\bm y
    \end{bmatrix} \right ) 
    \psi\left(\begin{bmatrix}
        \bm x'\\ \bm y'
    \end{bmatrix} \right ),
    \end{align*}
    which again by the definition of kernels, proves that $k_r(\cdot, \cdot)$ is a positive-definite kernel.
\end{proof}
An immediate consequence of Proposition \ref{prop2} is that $k_r$ is the reproducing kernel of a reproducing kernel Hilbert space (RKHS) $\mathcal{H}_{k_r}(\mathcal{\bar X})$. Also, this structure allows us to choose $k_r^i$'s specifically for each region $r \in \mathcal{I}$. 

We denote the input samples and domain by $(\bm x, \bm y) \in \mathcal{\bar X}$ as defined above. Let $X_r \in \mathbb{R}^{n \times N_r}$ and $Y_r \in \mathbb{R}^{(m+1) \times N_r}$ be matrices whose columns are vectors $\bm x_j$ and $\bm y_j$ of the corresponding dataset $\mathcal{D}_r$ and let outputs $\bm \omega_r \in \mathbb{R}^{N_r}$ be the stacked vector of the corresponding outputs to the states $\bm x_j \in \mathcal{R}_r$. Based on the collected batches, the MOGP model (\ref{eq11}) equipped with the kernels of the form $k_r$ in Proposition \ref{prop2}, gives the following expression for the posterior distribution of a test point $(\bm x_*, \bm y_*) \in \mathcal{\bar X}_r$:
\begin{align}
    m_r^* &= \bm \omega_r^T (k_r + \sigma_n^2 I )^{-1}\bar{K}_r^T\bm y_*,\label{eq14}\\
    {\sigma_r^*}^2&= \bm y_*^T ( \Lambda_r(\bm x_*,\bm x_*) - \bar{K}_r(K_r + \sigma_n^2 I)^{-1}\bar{K}_r^T) \bm y_*,\label{eq15}
\end{align}
where $K_r$ is the Gram matrix of $k_r(\cdot, \cdot)$ for the input data pair $(X_r,Y_r)$, and $\bar K_r \in \mathbb{R}^{(m+1)\times N_r}$ is given by
\begin{align*}
    \bar K_r &= \begin{bmatrix}
        \bm{\bar{k}}^1_r & \bm{\bar{k}}_r^2 & \dots &\bm{\bar{k}}_r^{N_r}
    \end{bmatrix} 
    \circ Y_r,\\
    \bm{\bar{k}}_r^i &= \begin{bmatrix}
        k_r^1(x_*, x_i)& \dots& k_r^{m+1}(x_*, x_i)
    \end{bmatrix}^T, i \in \{1, \dots, N_r \},
\end{align*}
where $\circ$ denotes the element-wise multiplication (Hadamard product) of two matrices with identical dimensions. The posterior mean and variance of $d_r^V$ and $d_r^V$ are obtained by substituting the output vectors $\bm \omega_r^V$ and $\bm \omega_r^h$ and their corresponding kernel structure into (\ref{eq14}) and (\ref{eq15}). 

This result highlights a key advantage of using the proposed kernel. The resulting expression for the posterior mean $m_r^* = \bm \mu_r(\bm x_*)\bm y_*$ and the variance ${\sigma_r^*}^2 =\bm y_*^T \Sigma_r(\bm x_*) \bm y_*$ are linear and quadratic in $\bm y_*$ (and the control input), respectively. We will exploit this feature in the next section to establish a convex GP-based safety filter.

\subsection{Confidence bounds for the estimation of the effect of uncertainty}

Although GP is inherently a probabilistic model, a high probability error bound can be derived for the distance between the true value and the GP prediction. This requires an additional assumption on $d_r^V$ and $d_r^h$. Let $d_r$ be either $d_r^V$ or $d_r^h$, and $\bm \varphi_r$ be either $\bm \varphi_r^V$ or $\bm \varphi_r^h$ from (\ref{eq13a}) and (\ref{eq13b}).

\begin{assum}\label{assum3}
    We assume that each $i^{th}$ element of $\bm \varphi_r$ is a member of $\mathcal{H}_{k_r^i}$ for $i \in \{1, \dots, m+1 \}$ with bounded RKHS norm $\| \varphi_r^i \|_{k_r^i} \leq \eta_r$ for $i \in \{1, \dots, m+1 \}$ within a given region $r \in \mathcal{I}$.
\end{assum}

\begin{lem}[\cite{srinivas2009gaussian}]
 Let Assumption \ref{assum3} hold. Then, with a probability of at least $1-\delta$, the following holds
\begin{equation} 
    \vert m_r^*-d_r(\bm x, \bm u)\vert \leq\beta_r\sigma_r^*, 
    \label{eq16}
\end{equation}
on $\mathcal{R}_r \subset \mathcal{X}, \delta \in (0,1)$ and $$\beta_r=\sqrt{2\eta_r^2+300 \kappa_{N_r+1} \log ^3((N_r+1) / \delta)},$$ where $\kappa_{N_r+1}$ is the maximum mutual information that can be obtained after getting $N_r+1$ data, and $\eta_r$ is the upper bound of the corresponding RKHS norm, and $m^*_r$ and $\sigma^*_r$ are the posterior mean and standard deviation of a test point $(\bm x_*, \bm y_*) \in \mathcal{\bar X}_r$.
\end{lem}

Based on the probabilistic bounds on (\ref{eq16}) and the CLF and CBF constraints (\ref{eq8a}) and (\ref{eq8b}), we can conclude that the following hold with a probability of at least $1 - \delta$
\begin{align}
    \dot V(\bm x, \bm u) &\leq \hat{\dot V}(\bm x, \bm u) + \sum^R_{r=1} \delta_r(m_r^V + \beta_r \sigma_r^V),\label{eq17}\\
    \dot h(\bm x, \bm u) &\geq \hat{\dot h}(\bm x, \bm u) +  \sum^R_{r=1} \delta_r( m_r^h - \beta_r \sigma_r^h),
    \label{eq18}
\end{align}
where $m_r^V(\bm x, \bm y)$ and $\sigma_r^V(\bm x, \bm y)$ are the mean and standard deviation of a query point derived from (\ref{eq14}), (\ref{eq15}). Incorporating (\ref{eq17}) and (\ref{eq18}) into the optimization problem (\ref{eq6}), we have
\begin{subequations}\label{eq19}
    \begin{align}
        \boldsymbol{u}_{s} ={} &\underset{\left( \bm u , d \right) \in \mathbb{R}^{m+1}}{\arg\min} \hspace{4pt} {\| \bm u \|^2_2 + \rho d^2}\label{eq19a}\\
        \textrm{s.t.}\quad&\hat{\dot V}(\bm x, \bm u) + \sum^R_{r=1} \delta_r(m_r^V + \beta_r \sigma_r^V) + \lambda V(\bm x) \leq d,\label{eq19b}\\
        &\hat{\dot h}(\bm x, \bm u) +  \sum^R_{r=1} \delta_r( m_r^h - \beta_r \sigma_r^h)  + \alpha(h(\bm x)) \geq 0, \label{eq19c}
    \end{align}
\end{subequations}
Note that the constraint in (\ref{eq19}) is constructed regardless of the underlying true dynamics.
\begin{thm}\label{thm1}
    The optimization problem (\ref{eq19}) is convex and can be converted into the standard SOCP of the form (\ref{eq20}), if $m_r^V, m_r^h$, and $\sigma_r^V, \sigma_r^h$ satisfy (\ref{eq14}) and (\ref{eq15}), respectively.
    \begin{align}
       \bm u_{socp} ={} &\underset{\bm z}{\arg\min} \hspace{4pt} \bm f^T \bm z \nonumber\\
        \text { s.t. } & \sum^R_{r=1} \delta_r \left\|M_r^i \bm z+\bm{n}_r^i\right\|_2 \leq \sum^R_{r=1} \delta_r ({\bm{p}_r^i}^T \bm z+q_r^i),
        \label{eq20}
    \end{align}
    for $i=1, \ldots, n_c$, where $\bm z \in \mathbb{R}^{n_z}$ and $\bm f \in \mathbb{R}^{n_z}$, and the matrix $M_r^i$ and $\bm{n}_r^i, \bm{p}_r^i, q_r^i$ have appropriate dimensions for $n_c$ constraints.
\end{thm}
\begin{proof}
    Let's denote the objective function of (\ref{eq19}) by $J=\left\|\bm u\right\|_2^2 + \rho d^2$. Let $\bm p = \begin{bmatrix}\bm u^T & d\end{bmatrix}^T \in \mathbb{R}^{m+1}$. Then, we can rewrite $J$ as the equivalent form of $J_1=\bm p^T Q \bm p$, where $Q$ is a positive definite diagonal matrix with matrix square root $E$. Then, we convert $J_1$ into $J_2 = \left\|E \bm p\right\|^2_2$. Since the Euclidean norm is always positive, we can consider the equivalent objective function $J_3=\left\|E \bm p\right\|_2$ and set $\left\|E \bm p \right\|_2 \leq t$, where $t \in \mathbb{R}$ is an auxiliary variable. Then, we convert $J_3$ to a second-order cone constraint (\ref{eq21}). Now, we need to solve a new minimization problem with a new augmented variable $\bm z = \begin{bmatrix}\bm u^T & d & t\end{bmatrix}^T\in \mathbb{R}^{m+2}$ as
    \begin{align}
        \min & \begin{bmatrix}\bm 0_{m+1}^T &1 \end{bmatrix} \bm z \nonumber\\
        \text { s.t. } & \| \underbrace{\begin{bmatrix}E & \bm 0_{m+1}\end{bmatrix}}_{\coloneqq M_r^1} \bm z\| \leqslant \underbrace{\begin{bmatrix} \bm 0_{m+1}^T & 1\end{bmatrix}}_{\coloneqq \bm {p_r^1}^T} \bm z
        \label{eq21},
    \end{align}
    where $\bm 0_{m+1} \in \mathbb{R}^{m+1}$ is a vector of zeros and $\bm f=  \begin{bmatrix}\bm 0_{m+1}^T &1 \end{bmatrix}^T$.

    Next, we need to show that constraints (\ref{eq19b}) and (\ref{eq19c}) are SOC constraints. We need to consider the constraints for a given region index $r \in \mathcal{I}$ as at each time step one region is active according to the defined switching signal $\delta_r$ in Section \ref{sec1}. As the proof is the same for both constraints, we consider the CBF constraint (\ref{eq19c}). Note that based on (\ref{eq14}), we have $m^h_r(\bm x, \bm y) = \bm \mu_r^h(\bm x) \bm y$, which can be written as
    \begin{align*}
        m^h_r(\bm x, \bm y) = \mu^{h1}_r + \bm \mu_r^{hm} \bm u,
    \end{align*}
    where the notation $\mu^{h1}_r$ refers the first element and $\bm \mu_r^{hm}$ refers the last $m$ elements of the row vector $\bm \mu^h_r \in \mathbb{R}^{1 \times (m+1)}$. Also, from (\ref{eq13b}), we know that $\bm \varphi_r^h$ is control affine, i.e. $\bm \varphi_r^h = \varphi_r^{hf} + \bm \varphi_r^{hg} \bm u$, where $\varphi_r^{hf} \in \mathbb{R}$ and the row vectors $\bm \varphi_r^{hg} \in \mathbb{R}^{1 \times m}$. Hence, the right hand side of the inequality in (\ref{eq19c}) is affine in $\bm u$ as desired.

    Based on (\ref{eq15}), we have that ${\sigma_r^h}^2(\bm x, \bm y) = \bm y^T \Sigma_r^h (\bm x) \bm y$, where $\Sigma_r^h \in \mathbb{R}^{(m+1) \times (m+1)}$. Since $k_r$ is a valid kernel and $\sigma_n>0$, $\Sigma_r^h$ is positive definite. Then, we have
    \begin{equation*}
         \sigma_r^h(\bm x, \bm y)=\sqrt{\bm y^T \Sigma_r^h \bm y}=\sqrt{\bm y^T L_r^T L_r \bm y}=\|L_r \bm y\|_2,
    \end{equation*}
    where $L_r \in \mathbb{R}^{(m+1) \times (m+1)}$ is the matrix square root of $\Sigma_r^h$. By the definition $\bm y = \begin{bmatrix}1 & \bm u^T\end{bmatrix}^T$, we have $\sigma_r^h(\bm x, \bm y) = \left\|\bm l_r^1+ L_r^m \bm u\right\|_2$, where the notation $\bm l_r^1$, and $L_r^m$ refers to the first column and the last $m$ columns of the matrix $L_r$, respectively. Now, we can rewrite the safety certificate (\ref{eq19c}) as a SOC constraint of the form
    \begin{equation}
        \left\|A_r^h(\bm x) \bm u+ \bm{b}_r^h(\bm x)\right\|_2 \leq \bm{c}_r^h(\bm x) \bm u+d_r^h(\bm x), 
        \label{eq22}
    \end{equation}
    where 
    \begin{align}
        A_r^h(\bm x) &= \beta_r L_r^m \in \mathbb{R}^{(m+1) \times (m)},\nonumber\\
        \bm b_r^h(\bm x) &= \beta_r \bm l_r^1 \in \mathbb{R}^{m+1}, \nonumber\\
        \bm c_r^h(\bm x) &= \bm \varphi_r^{hg} + \bm \mu_r^{hm} \in \mathbb{R}^{1 \times m}\nonumber\\
        d_r^h(\bm x) &= \varphi_r^{hf} + \mu_r^{h1} \in \mathbb{R}.\label{eq26p}
    \end{align}
    
    It can be expressed in terms of the new variable $\bm z$ as
    \begin{equation*}
        \| \underbrace{\begin{bmatrix}
        A_r^h(\bm x) & \bm{0}_{m+1}
        \end{bmatrix}}_{\coloneqq M_r^2} \bm z+\underbrace{\bm{b}_r^h(\bm x)}_{\coloneqq \bm n_r^2} \|_2 \leq \underbrace{\begin{bmatrix}\bm{c}_r^h(\bm x) & 0 \end{bmatrix}}_{\coloneqq \bm {p_r^2}^T} \bm z+ \underbrace{d_r^h(\bm x)}_{\coloneqq q_r^2}.
    \end{equation*}
The same proof holds for the CLF constraint (\ref{eq19b}). Adding the resulting constraints to (\ref{eq21}) leads to (\ref{eq20}), which concludes the proof.
\end{proof}

\subsection{Feasibility analysis}\label{sec3d}
In this section, we will analyze the point-wise feasibility of the SOCP (\ref{eq20}). As the CLF is considered as a soft constraint, only the CBF constraint restricts the feasibility. Intuitively, if the batch MOGP prediction is less accurate, the resulting standard deviation $\sigma_r^h(\bm x, \bm y)$ will be large, leading to a more conservative approach. In particular, it restricts the space on which a safe control input can be selected, which may make the SOCP infeasible. In the following, we theoretically analyze the conditions for feasibility.
\begin{thm}\label{thm3}
    Given a state $\bm x \in \mathcal{R}_r$, the SOCP (\ref{eq20}) is feasible if and only if there exists a control input $\bm u \in \mathbb{R}^m$ that satisfies the following conditions
        \begin{align}
               & \begin{bmatrix}
                d_r^h(\bm x) & \bm{c}_r^h(\bm x)
                \end{bmatrix}
                \begin{bmatrix}
                    1\\ \bm u
                \end{bmatrix}
                \geq 0,\label{eq23}\\
            &\begin{bmatrix}
                1 & \bm u^T
            \end{bmatrix} S_r(\bm x) \begin{bmatrix}
                1 \\ \bm u
            \end{bmatrix} \leq 0,
            \label{eq24}
        \end{align}
        where $S_r \in \mathbb{R}^{(m+1) \times (m+1)}$ is of the form
        \begin{align}
            S_r(\bm x) &= \begin{bmatrix}
            S_r^1(\bm x) & S_r^2(\bm x) \\
            {S_r^2}^T(\bm x) & S_r^3(\bm x)
            \end{bmatrix},\nonumber\\
            S_r^1(\bm x) &= {\bm{b}_{r}^h}^T(\bm x) {\bm{b}_{r}^h}(\bm x) - {d_r^h}^T(\bm x)d_r^h(\bm x),\nonumber\\
            S_r^2(\bm x) &= {\bm b_r^h}^T(\bm x) A_r^h(\bm x)-{d_r^h}^T(\bm x) \bm c_r^h(\bm x),\nonumber\\
            S_r^3(\bm x) &= A_r^h(\bm x)^T A_r^h(\bm x)-\bm c_r^h(\bm x)^T \bm c_r^h(\bm x),
            \label{eq25}
        \end{align}
        with $A_r^h(\bm x), \bm b_r^h(\bm x), \bm c_r^h(\bm x)$, and $d_r^h(\bm x)$ are defined in (\ref{eq26p}).
\end{thm}

\begin{proof}
    In order to analyze the feasibility of the SOCP (\ref{eq20}), we need to verify that the safety SOC constraint (\ref{eq22}) is feasible. Since the left hand side is always non-negative, the right hand side must be also non-negative, which leads to the first condition.

    Now, since both sides of (\ref{eq22}) are non-negative, we can take squares of both sides 
    and write the resulting expression in quadratic form.
    and collect all the terms on the left.
    Factorizing the similar terms gives the second condition.
\end{proof}
Based on this result, we obtain a necessary condition for point-wise feasibility as stated in the following:
\begin{cor}\label{cor2}
    Given $\bm x \in \mathcal{R}_r$, if the SOCP (\ref{eq20}) is feasible, then the following condition must be satisfied
    \begin{equation}
        1 - \frac{1}{\beta_r^2} \bm{\phi}_r {\Sigma_r^h}^{-1} \bm{\phi}_r^T \leq 0,\label{eq26}
    \end{equation}
    where $\bm{\phi}_r = \begin{bmatrix}
    \varphi_r^{hf} + {\mu}_r^{h1} & \bm{c}_r^h(\bm{x})\end{bmatrix} \in \mathbb{R}^{1 \times (m+1)}$.
\end{cor}
\begin{proof}
We prove this by contradiction. Let's assume that there exists a solution $\bm u \in \mathbb{R}^m$ of the SOCP (\ref{eq20}), which satisfies $1 - \frac{1}{\beta_r^2} \bm{\phi_r} {\Sigma_r^h}^{-1} \bm{\phi}_r^T > 0$. 
Let's define the matrix $M_r$ partitioned as
\begin{equation*}
    M_r = \begin{bmatrix}
        1 & \bm{\phi}_r \\
        \bm{\phi}_r^T & \beta_r^2 \Sigma_r^h(\bm{x})
    \end{bmatrix} \in \mathbb{R}^{(m+2) \times (m+2)}.
\end{equation*}
The Schur complements of $M_r$ is obtained by
\begin{equation*}
    \left\{
    \begin{array}{l}
        M_r / 1 = \beta_r^2 \Sigma_r^h(\bm{x}) - \bm{\phi}_r^T \bm{\phi}_r, \\
        M_r / \beta_r^2 \Sigma_r^h(\bm{x}) = 1 - \frac{1}{\beta_r^2} \bm{\phi}_r {\Sigma_r^h}^{-1} \bm{\phi}_r^T.
    \end{array}
    \right.
\end{equation*}
We can easily verify that $\beta_r^2 \Sigma_r^h=\begin{bmatrix}\beta_r \bm l_r^1 & A_r^h(\bm x)\end{bmatrix}^T\begin{bmatrix}\beta_r \bm l_r^1 & A_r^h(\bm x)\end{bmatrix}$ and $M_r / 1=S_r(\bm x)$.
We know that $\Sigma_r^h(\bm x)$ is symmetric positive-definite (SPD). So, $M_r$ is also symmetric and we can characterize its definiteness by Schur complement theorem. 
Since its lower right block is positive definite and its corresponding Schur complement $M_r/\beta_r^2 \Sigma_r^h(\bm x)>0$ by our assumption, based on \cite{boyd2004convex}, we can conclude that $M_r$ is also positive definite.

Now, given that $M_r$ is SPD, and its upper left block is always positive definite, we can conclude that $M_r/1 = S_r(\bm x)$ must be positive definite, which is a contradiction to the condition (\ref{eq24}). Thus, $1 - \frac{1}{\beta_r^2} \bm{\phi}_r {\Sigma_r^h}^{-1} \bm{\phi}_r^T \leq 0$ holds.
\end{proof}

Theorem $\ref{thm3}$ also provides the theoretical background to obtain a sufficient condition of feasibility.
\begin{cor}\label{cor3}
    Given $\bm x \in \mathcal{R}_r$, the SOCP (\ref{eq20}) is feasible at $\bm x$, if $S_r^3(\bm x)$ is negative definite.
\end{cor}
\begin{proof}
    From (\ref{eq25}), we know that $S_r^3(\bm x)$ is a symmetric matrix, thus it has real eigenvalues. Let's denote its maximum eigenvalue and the corresponding eigenvector by $\lambda_r^{m}$, $\bm{e}_r^{m}$. Then, by the assumption that $S_r^3(\bm x)$ is negative definite, we can conclude that $\lambda_r^{m}<0$ and ${\bm{e}_r^{m}}^T S_r^3(\bm x) \bm{e}_r^{m}<0$. Substituting (\ref{eq25}) into the former, we have
    \begin{equation*}
        {\bm{e}_r^{m}}^T {A_r^h}^T(\bm x) A_r^h(\bm x) \bm{e}_r^{m} - (\bm c_r^h(\bm x) \bm{e}_r^{m})^T(\bm c_r^h(\bm x)\bm{e}_r^{m})<0.
    \end{equation*}
    Since ${A_r^h(\bm x)}^TA_r^h(\bm x)$ is positive definite, $\bm c_r^h(\bm x) \bm{e}_r^{m} \neq 0$ must hold. We take a control input in the direction of this eigenvector as $\bm u_e = \alpha \sgn(\bm{c}_r^h(\bm x) \bm{e}_r^{m})\cdot\bm{e}_r^{m}, \alpha>0$. Next, we need to check if the resulting control input can satisfy the conditions of Theorem \ref{thm3}. Substituting $\bm{u}_e$ in (\ref{eq24}), we have
    \begin{equation*}
        S_r^1 + 2\alpha S_r^2 ((\bm{c}_r^h(\bm x) \bm{e}_r^{m})\cdot\bm{e}_r^{m}) + \alpha^2 {\bm e_r^{m}}^T S_r^3(\bm x) \bm e_r^{m}.
    \end{equation*}
    By choosing large enough $\alpha$, the above statement can be made negative since ${\bm e_r^m}^T S_3(\bm x) \bm e_r^m <0$.
    Also, substituting $\bm u_e$ into (\ref{eq23}), we have
    \begin{equation*}
        \bm c_r^h(\bm x)  \bm u_e + (\varphi_r^{hf} + \mu_r^{h1}) = \alpha \lvert \bm c_r^h(\bm x) \bm e_r^m \rvert + (\varphi_r^{hf} + \mu_r^{h1}).
    \end{equation*}
    Again, by choosing a large enough $\alpha$, we can make the above expression positive, which concludes the proof.
\end{proof}
\subsection{Data collection and model fitting}\label{sec3e}
We use an episodic data collection method. In the first episode, we run an initial roll-out using the CLF and CBF designed based on the nominal system. Using nominal system (\ref{eq7}), we run QP controller (\ref{eq6}) with the nominal values of the time derivative of $V(\bm x)$ and $h(\bm x)$. Then, we record the system's trajectory until it reaches an unsafe state. Next, we partition the resulting dataset into $R$ datasets $\mathcal{D}_r$ as described in Section \ref{sec3b}. Subsequently, we fit MOGPs configured with the kernel functions $k_r$ to batches $\mathcal{D}_r$, for $r=1, \dots, R$. To use (\ref{eq14}) and (\ref{eq15}) for prediction, we need to infer the hyperparameters $\bm \theta_r$ for each individual kernel $k_r^i, i\in\{1, \dots,m+1\}$ in the structure of $k_r$. Each set of parameters is obtained by minimizing the negative log marginal likelihood \cite{liu2018remarks} as
\begin{equation}
    \bm{\theta}_{r}^{opt} ={} \underset{\bm \theta_r}{\arg\min} \hspace{4pt} {-\log p({\bm \omega}_r\mid\bar X_r, \bm \theta_r)},
    \label{eq27}
\end{equation}
where for the $J_{\bm \theta_r} = -\log p({\bm \omega}_r\mid\bar X_r, \bm \theta_r)$, we have
\begin{align*}
    J_{\bm \theta_r} = &\frac{1}{2} \mathbf{\bm \omega}_r^{\mathrm{T}}\left[K_r(\bar X_r, \bar X_r)+\sigma_n^2 I\right]^{-1} \mathbf{\bm \omega}_r\\
    &+\frac{1}{2} \log \left|K_r(\bar X_r, \bar X_r)+\sigma_n^2 I\right|+\frac{N_r}{2} \log 2 \pi,
\end{align*}
where $\bar X_r = (X_r, Y_r)$ such that $((X_r, Y_r), \bm \omega_r) \in \mathcal{D}_r$. $\bm \omega_r$ is either the vector of outputs ${\bm \omega}_j^V$ or ${\bm \omega}_j^h$ for approximating $d_r^V$ or $d_r^h$, respectively.
Then we use the uncertainty-aware SOCP optimization (\ref{eq19}) with a high probability bound ($1 - \delta = 0.95$) and run the system with the proposed method until the system reaches an unsafe state or the problem becomes infeasible. The data collected during each episode will be added to the dataset and the hyperparameters $\bm \theta_r$ will be optimized using (\ref{eq27}) with the new dataset. We repeat this process until the simulation completed without encountering any unsafe behavior.
\section{Simulation results}
In this section, we highlight the effectiveness of the proposed method by applying it to a switching adaptive cruise control (ACC) system that is moving on varying road conditions. For comparison, we also implemented a baseline GP-CBF-CLF-SOCP \cite{castaneda2021pointwise} that fits a single GP to piecewise residuals and a nominal QP controller. The GP models were trained using GPyTorch \cite{gardner2018gpytorch}.

Consider a switching ACC system modeled by 
\begin{align}
    \dot {\bm{x}} &= \sum^R_{r=1} \delta_r \left( f_r(\bm{x})+ g_r(\bm{x})\right) u,\nonumber\\ 
    f_r(\bm x) &= \begin{bmatrix}
        x_2 \\ -\frac{F_r(x_2)}{m}
    \end{bmatrix}, g_r(\bm x) = \begin{bmatrix}
        0 \\ \frac{c_r}{m}
    \end{bmatrix},
    \label{eq28}
\end{align}
where the system state $\bm x = \begin{bmatrix} x_1 & x_2\end{bmatrix}^T \in \mathcal{X}$ containing the position of the ego car $x_1$, and its forward velocity $x_2$ in the state space $\mathcal{X} = [0,700] \times [0,30]$. The control input is the wheel force $u \in \mathbb{R}$. The mass of the ego car is denoted by $m$ and the rolling resistance is modeled by $F_r(x_2) = f_r^0+f_r^1 x_2+ f_r^2 x_2^2$.

The switching behavior of the system arises from the changing road conditions which impact the rolling resistance $F_r(x_2)$ and traction control of the vehicle $c_r$. As the vehicle encounters different road surfaces $F_r(x_2)$ and $c_r$ will change, leading to variations in the system dynamics. We consider hyper-rectangular partitions of $\mathcal{X}$ with bounds $\mathcal{R}_1 = ([0,15)\cup(25,700])\times [0,30]$, and $\mathcal{R}_2 = [15,25]\times [0,30]$ that represent two different road conditions. The state-dependent switching signal $\delta_r$ in (\ref{eq1}) determines the active region index at each time step. We only have access to the nominal model and the true system is unknown in both operating regions. The parameters for each model are shown in Table \ref{tab1}.

\begin{table}[htbp]
\centering
\caption{\small Parameters for the nominal and true system in $\mathcal{R}_1, \mathcal{R}_2$}
\begin{tabular}{|l|l|l|l|l|l|}
\hline
 &$m(kg)$ & $f_1^0$ & $f_1^1$ & $f_1^2$ & $c_r$ \\
\hline
True $\mathcal{R}_1$ & 3300 & 0.2 & 10 & 0.5 & 1 \\
True $\mathcal{R}_2$ & 3300 & 1 & 50 & 4.5 & 0.5 \\
Nominal & 1050 & 0.1 & 15 & 2.25 & 1 \\
\hline
\end{tabular}\label{tab1}
\end{table}

The distance between the ego car and the front car, denoted by $z$, is governed by the equation $\dot z = v_0 - x_2$, where $v_0 = 10\hspace{2pt}m/s$ represents the constant velocity of the front car. The control objective is to achieve a target speed of $v_d$ while ensuring a safe distance from the front vehicle. To stabilize the velocity at the desired value, a Lyapunov function $V(\bm x) = (x_2 - v_d)^2$ is considered. Additionally, a CBF $h(\bm x) = z - T_h x_2$, with $T_h = 1.6$, is utilized to maintain a safe distance proportional to the current velocity.

We selected infinitely differentiable squared exponential kernel \cite{williams2006gaussian}, for $k_i, i=1,2$ base kernels in the configuration of $k_r$. We used the process in Section \ref{sec3e} to collect the dataset and optimize the batch MOGP model hyperparameters, which collected a total number of $390$ data samples in the initial roll-out and a total of $999$ data samples in all episodes.
The simulation is started from the initial state $\bm x_0 = \begin{bmatrix}0 & 14\end{bmatrix}^T$ and $z_0 = 140\hspace{2pt}m$. The ego car starts in the region $\mathcal{R}_1$ from a safe distance behind the front car. First, the controller increases the velocity to converge to the desired speed resulting in a decrease in $z$. Then, the ego car transits into region $\mathcal{R}_2$ and approaches the front car which activates the safety constraint (\ref{eq19c}). Consequently, the controller must reduce the velocity to keep the safe distance. Due to the adverse effect of the piecewise residuals, the Single GP-SOCP and nominal QP controller cannot maintain the safe distance. 
However, the proposed MOGP-SOCP controller effectively avoids this unsafe behavior by leveraging the batch MOGP design. A simulation video showing the performance of the nominal and MOGP-SOCP controller is provided. Figure \ref{fig1} shows a snapshot of a simulation when all controllers result in the closest distance to the front car. It can be verified that only MOGP-SOCP controllers can keep the safe distance during the simulation.



\begin{figure}
\minipage{0.33\linewidth}
  \includegraphics[width=\linewidth]{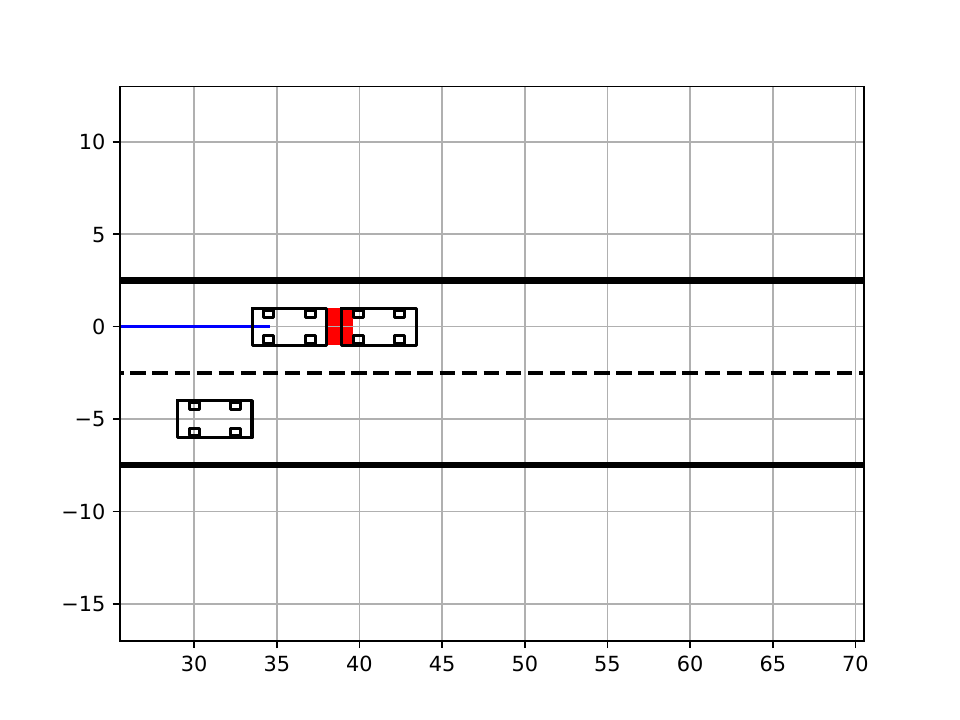}
\endminipage\hfill
\minipage{0.33\linewidth}
  \includegraphics[width=\linewidth]{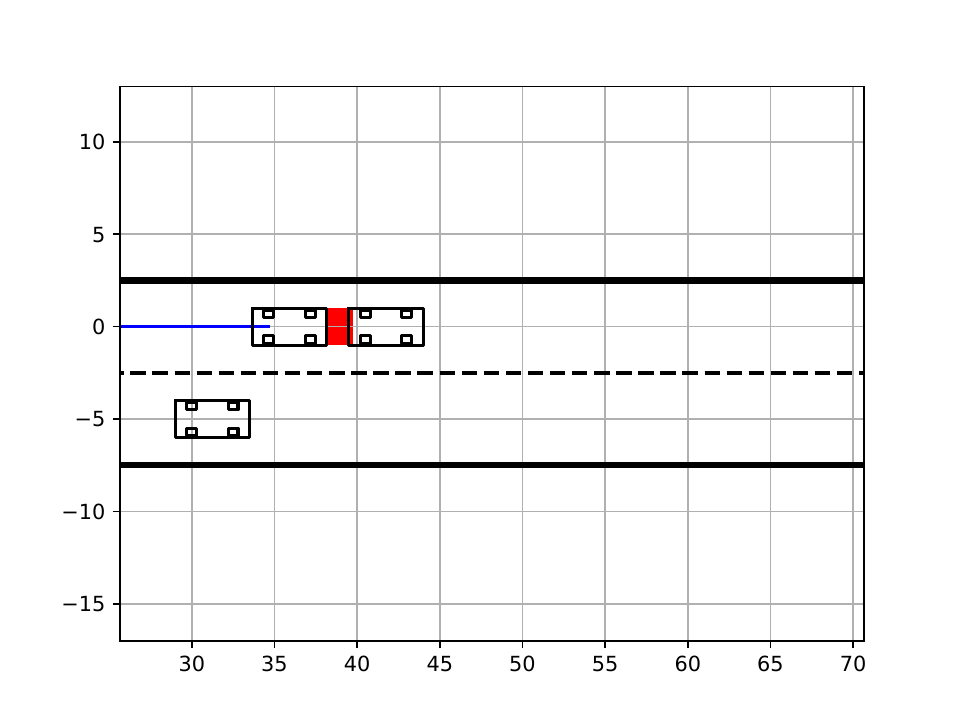}
\endminipage\hfill
\minipage{0.33\linewidth}%
  \includegraphics[width=\linewidth]{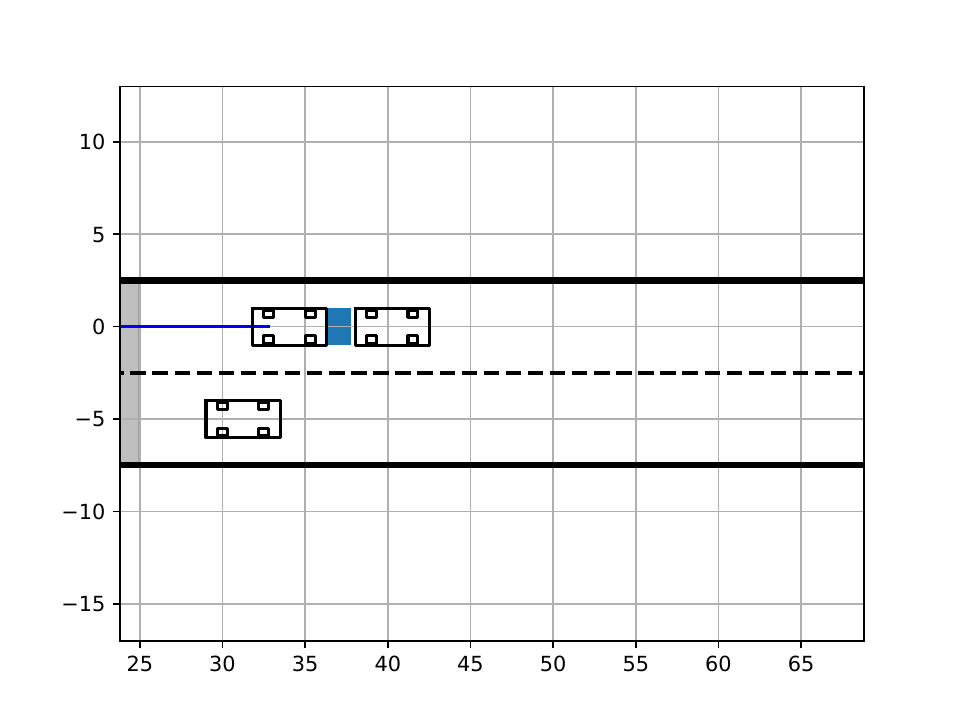}
\endminipage
  \caption{\footnotesize Snapshots of the simulation to show the unsafe (red color) and safe (green color) distance between cars in different implementations. The nominal QP controller (left) and single GP-SOCP controller (middle) violate the safety distance. MOGP-based controller (right) ensures safety. See https://youtu.be/8nmcIIJSGJE for the simulation video.}\label{fig1}
\end{figure}

In Figure \ref{fig2}, we aim to compare the performance of the proposed method with the baseline and nominal controllers. For your reference, we provided the trajectories of the oracle true design. It has been illustrated in Figure \ref{fig2} that the single GP-SOCP controller did not converge to the desired velocity, while MOGP-SOCP and nominal QP controllers could reach $v_d$ in the first $4$ seconds. Then, the ego car transited to $\mathcal{R}_2$ at $t = 6\hspace{2pt}s$ which affects the rate of change of the velocity. As it approaches the front car, the controllers must decrease the velocity in response to remain safe. The CBF $h(\bm x)$ and control input $u$ are also illustrated in this figure. This plot shows that the single GP-SOCP and nominal QP controllers violate the safety condition approximately at $t = 15\hspace{2pt}s$ and $t = 17\hspace{2pt}s$, respectively. However, the MOGP-SOCP controller could maintain the safe distance to the front car during the simulation. Also, it can be verified from the control input trajectory that the MOGP-SOCP controller could recover the true system's performance. However, the single GP-SOCP controller does not generate a smooth control input which is a result of poor uncertainty quantification between the switching times.

 \begin{figure}[t]
    \centering
    \includegraphics[width=0.8\columnwidth]{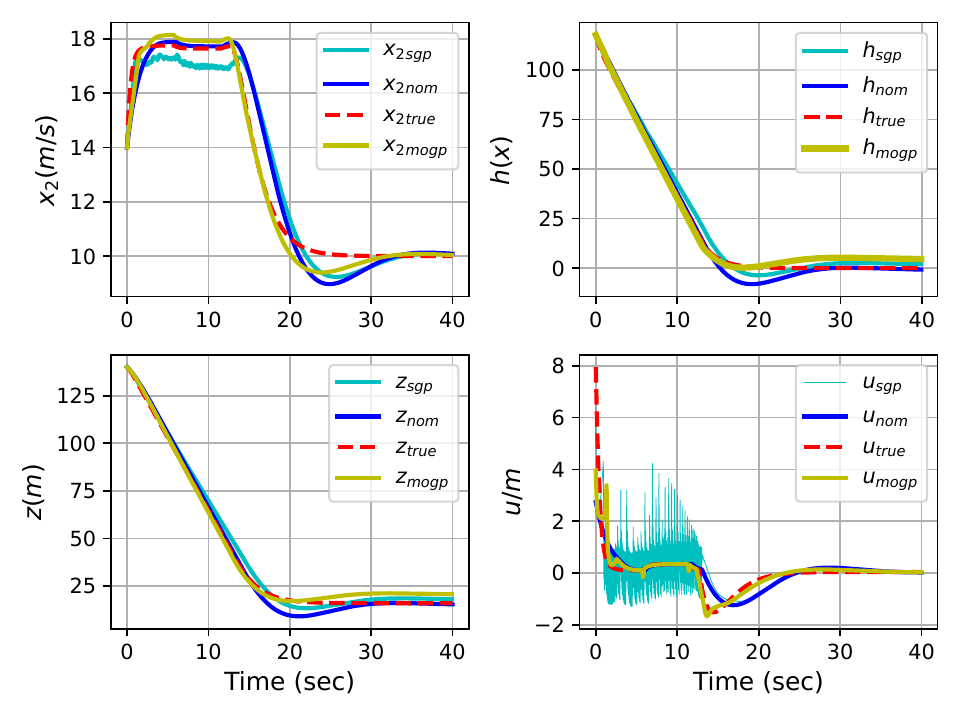}
    \caption{\footnotesize Comparison of the MOGP-SOCP (yellow), single GP-SOCP (cyan), nominal QP (blue), and true design (red dashed). System state $x_2$ and distance between two cars $z$ (left column). The CBF and control input $u$ (right column).}
    \label{fig2}
\end{figure}

\section{Conclusion}

In this paper, A batch MOGP framework is developed to approximate the effect of the uncertainty on the CLF and CBF constraints. The switching dynamics of the true system result in piecewise residuals in CLF and CBF constraints corresponding to $R$ regions that cover the state space.
A batch MOGP model is designed to capture the piecewise residuals in each region. Then, the resulting constrained optimization problem with the uncertainty-aware chance constraint is converted into a SOCP. This optimization is proven to be convex and can be solved in real-time. Finally, the feasibility of the resulting SOCP is addressed.
Future works will focus on eliminating the assumption of no impulse effects and non-overlapping regions in the switching system.


\bibliographystyle{plain}
\bibliography{References}
\end{document}